\documentclass{article}

\usepackage{amsmath,amssymb,amsthm}
\usepackage{multirow}
\usepackage{url}

\newcommand{\vertiii}[1]{{\vert\kern-0.25ex\vert\kern-0.25ex\vert #1 
    \vert\kern-0.25ex\vert\kern-0.25ex\vert}}

\usepackage{natbib} 
\usepackage[hang,small,bf]{caption}
\usepackage{graphicx}

\usepackage{xcolor}

\usepackage[CJKbookmarks=true,
            bookmarksnumbered=true,  
			bookmarksopen=true,
			colorlinks=true,
			citecolor=blue,
			linkcolor=blue,
			anchorcolor=red,
			urlcolor=blue]{hyperref}
\usepackage[capitalize,noabbrev]{cleveref}
\usepackage{todonotes}
\usepackage{booktabs} 
\usepackage{subcaption}
\usepackage{multirow}
\usepackage{accents}
\newlength{\dhatheight}

\usepackage[margin=1in]{geometry} 

\newcounter{rcnt}[section]

\def\argmin{\mathop{\rm argmin}}
\def\argmax{\mathop{\rm argmax}}

\newcommand{\pinull}{\pi_{\mathrm{null}}}
\newcommand{\hpinull}{\hat\pi_{\mathrm{null}}}

\newcommand{\EE}{\mathbb{E}}
\newcommand{\E}{\mathbb{E}}

\newcommand{\III}[1]{{\left\vert\kern-0.25ex\left\vert\kern-0.25ex\left\vert #1 
    \right\vert\kern-0.25ex\right\vert\kern-0.25ex\right\vert}}

\makeatletter
\providecommand*{\diff}%
	{\@ifnextchar^{\DIfF}{\DIfF^{}}}
\def\DIfF^#1{%
	\mathop{\mathrm{\mathstrut d}}%
		\nolimits^{#1}\gobblespace}
\def\gobblespace{%
	\futurelet\diffarg\opspace}
\def\opspace{%
	\let\DiffSpace\!%
	\ifx\diffarg(%
		\let\DiffSpace\relax
	\else
		\ifx\diffarg[%
			\let\DiffSpace\relax
	\else
		\ifx\diffarg\{%
			\let\DiffSpace\relax
		\fi\fi\fi\DiffSpace}

\theoremstyle{plain}
\newtheorem{theorem}{Theorem}

\newtheorem{lemma}[theorem]{Lemma}

\theoremstyle{definition}

\DeclareMathOperator{\lfdr}{lfdr}

\theoremstyle{remark}

\newcommand{\BS}[1]{
\text{BS}{(#1)}
}

\usepackage[blocks, affil-it]{authblk}
\title{Calibration without labels in multiple testing}
\author{Adway S. Wadekar}
\author{Jake A. Soloff}
\affil{Department of Statistics, University of Michigan}
\date{\today}

\begin{document}

\maketitle

\begin{abstract}
    Large-scale hypothesis testing supports probability claims about individual hypotheses, as in empirical Bayes methods for estimating local false discovery rates.
We study how such claims can be interpreted as approximately calibrated forecasts of the null hypothesis, yielding interpretable error probabilities even under model misspecification.
Our approach draws conceptual inspiration from probabilistic forecasting but addresses a different challenge: unlike forecasting, where labels are eventually observed, in multiple testing the ground truth is never revealed, so calibration must be assessed stochastically and established indirectly.
We address this challenge by constructing a set of pseudo-labels, derived from the spacings of ordered $p$-values, which have the local false discovery rate as their regression target.
Our construction unlocks existing tools for assessing and performing post-hoc calibration in multiple testing.
Notably, we find on a large-scale empirical survey of published psychology and neuroscience literature that the $q$-value, a popular error measure based on the false discovery rate, can be severely miscalibrated.
\end{abstract}

\section{Introduction}\label{sec:intro}

Multiple testing is often presented as a problem of \emph{multiplicity correction}, treating the proliferation of hypothesis tests as a liability to be managed. 
With more hypotheses come more opportunities for spurious rejection, and the central statistical task is commonly framed as controlling the rate at which such false discoveries occur \citep{benjamini1995controlling}. 
A complementary tradition, originating with Robbins' empirical Bayes~\citep{robbins1951asymptotically,robbins1956empirical,robbins1963empirical,robbins1964empirical} and developed extensively by Efron and others~\citep{efron2007size,efron2008microarrays,efron2010LSI,efron2001empirical,efron2002empirical}, has long recognized that a large collection of related tests is a \emph{resource}: the empirical distribution of test statistics carries structural information that enables estimation of precise, hypothesis-specific error probabilities.

We operate within this empirical Bayes tradition but propose \emph{calibration} as a key inferential goal of multiple testing.
Rather than treating posterior probabilities as endpoints to be estimated under a specific model, we ask when a reported score~$g(p)\in [0,1]$ stands on its own as a meaningful error rate: among tests assigned a score of $g(p)$, the corresponding null hypothesis is true at about that rate.
Because calibration is a property of the score itself, rather than of the model used to derive it, it offers a route to meaningful inference about individual hypotheses even under model misspecification.
We formalize this perspective by reframing multiple testing as a binary classification problem with hidden labels \citep{genovese2002operating,genovese2004stochastic}.
Each test has an unobserved, binary label~$Y_i\in \{0,1\}$ indicating whether the null hypothesis is true, and the corresponding $p$-value $p_i$ plays the role of a covariate. 
From this perspective, the local false discovery rate~\citep{efron2001empirical} is simply the regression function
\begin{align}\label{eq:lfdr}
    \lfdr(p) = \mathbb{P}\{Y=1\mid p\}.
\end{align}
Unlike standard supervised classification, in multiple testing, the labels $(Y_i)$ are unobserved, but one of the class-conditional densities is known: under the null, we assume 
\begin{align}\label{eq:uniform-nulls}
    p\mid Y=1\sim \text{Unif}([0,1]).
\end{align}
Thus, multiple testing might be thought of as \emph{stochastically supervised classification}, in which partial information about the data generating process substitutes for labeled data.

This framing clarifies why calibration is an appealing goal for hypothesis testing. 
A map $g : [0,1]\to [0,1]$ is calibrated if $\mathbb{E}[Y\mid g(p)] = g(p)$: conditional on receiving the score~$g(p)$, the null is true at a rate exactly~$g(p)$.
Calibration allows practitioners to interpret a score as a meaningful probability without further qualification. 
It has been studied extensively across meteorology~\citep{hallenbeck1920forecasting,glenn1950verification,murphy1973new}, statistics~\citep{dawid1982well, degroot1983comparison,gneiting2007probabilistic,gneiting2007strictly}, machine learning~\citep{platt1999probabilistic,guo2017calibration,gupta2020distribution}, game theory~\citep{foster1998asymptotic,foster1999proof,hart2000simple}, and algorithmic fairness~\citep{hebert-johnson18a,kleinberg2017inherent}.
Despite its broad history, calibration has received little attention in multiple testing, perhaps because the labels are unobserved and so calibration ostensibly cannot be assessed with standard tools. 
Yet the central insight of large-scale inference is precisely that many related experiments create new opportunities for statistical learning.
We show that calibration becomes statistically accessible with many $p$-values, even without ground truth labels.

\paragraph{Our contributions.} We highlight three main contributions:
\begin{enumerate}
    \item We propose a \textbf{pseudo-label construction} from $p$-value spacings whose regression target is the $\lfdr$, allowing tools from supervised calibration to be directly applied to multiple testing.
    \item We study isotonic calibration applied to the pseudo-labels and show it coincides with nonparametric empirical Bayes estimation of the $\lfdr$~\citep{soloff2024edge,strimmer2008unified}. We prove a finite-sample bound on the Brier regret that vanishes as~$m\to\infty$ even when the true $\lfdr$ is not monotone.
    \item We adapt reliability diagrams and calibration-error diagnostics to our setting and use them to demonstrate, in simulations and on real data, that standard FDR-based error measures can be severely miscalibrated, in some cases even worse than raw~$p$-values.
\end{enumerate}

\paragraph{Related work.} Our framework builds on the findings of \citet{xiang2025frequentist}, who observe that the $\lfdr$ provides a perfectly calibrated forecast of the null hypothesis. 
We develop this perspective by considering approximate calibration and how to properly define, assess and control calibration error.
\citet{panagiotou2012should} essentially used calibration (without calling it that) as a diagnostic for selecting a significance threshold, asking whether discoveries near the decision boundary replicate at some pre-specified rate. Our use of calibration is more explicit and broader: rather than treating it as a sanity check on a threshold, we view post-hoc calibration as a key inferential goal of multiple testing.

The methodology in Section~\ref{sec:post-hoc} builds on existing nonparametric approaches to estimating the $\lfdr$ under monotonicity \citep{soloff2024edge,strimmer2008unified}. We give alternative interpretations of this method and further justify its application even when monotonicity is violated. We also formulate our methodology as fitting a regression on the pseudo-labels, which at a high level relates to prior work on placing modeling assumptions directly on the lfdr rather than on the marginal density of the $p$-values \citep{klaus2011learning,rice2008comment}.

\section{Framework: calibration in multiple testing}
\label{sec:framework}

We now formalize the connection to binary classification introduced in Section~\ref{sec:intro}. 
For each experiment $i=1,\ldots,m$, let $Y_i = 1$ (or $0$) if the $i^\text{th}$ null hypothesis is true (or false).\footnote{In multiple testing we typically write $H_i=0$ if the null hypothesis is true, so the labels in this paper are given by~$Y_i=1-H_i$. In this paper, we flip the labels for notational convenience and to further highlight the connection to binary classification.} 
We further assume that the pairs $(p_i, Y_i)$ follow a two-groups model~\citep{efron2001empirical}, meaning they are independent and identically distributed across $i=1,\ldots, m$. Let $(p, Y)$ denote a generic copy, let $\pinull = \mathbb{P}\{Y = 1\}$ denote the overall prevalence of nulls, and let~$f$ (and $F$) denote the marginal density (and cdf) of~$p$. 
In this section, we introduce the two key ingredients for calibration in multiple testing. We first review notions of perfect and approximate calibration. Second, we construct a set of pseudo-labels, derived from the $p$-value spacings, which allow these concepts borrowed from supervised calibration to be utilized from $p$-values alone.

\subsection{Calibration, coarsening, and Brier regret}

A \emph{predictor} is any measurable function $g:[0,1]\to [0,1]$. We say that $g$ is \emph{perfectly calibrated} if 
\begin{equation}\label{eq:calib-def}
    \mathbb{E}[Y \mid g(p)] = g(p) \qquad \text{almost surely}.
\end{equation}
A score of $g(p) = 0.1$ thus guarantees that the null hypothesis is true $10\%$ of the time among experiments receiving that score. By the law of iterated expectations, \eqref{eq:calib-def} is equivalent to 
\begin{equation}\label{eq:calib-characterization}
    g(p) = \EE[\lfdr(p) \mid g(p)] \qquad \text{almost surely},
\end{equation}
so every perfectly calibrated predictor is a coarsening of the~$\lfdr$ \citep{gupta2020distribution,xiang2025frequentist}. 
The $\lfdr$ is thus the finest perfectly calibrated predictor; at the opposite extreme, the constant function~$g_0(p)\equiv\pinull$ is the coarsest, since it ignores the data~$p$ entirely. 

In multiple testing, it is natural to impose the additional constraint that smaller $p$-values should represent stronger evidence against the null.
This is not an assumption on the data generating process, but rather a design decision that amounts to restricting~$g$ to the class of monotone predictors:
\[
\mathcal{G}_\uparrow = \left\{ g : [0,1]\to [0,1] : g \text{ is nondecreasing}\right\}.
\]
The monotone analogue of the $\lfdr$ is the \emph{isotonized lfdr}, the $L_2(F)$ projection of the $\lfdr$ onto~$\mathcal{G}_\uparrow$:
\[
\lfdr_\uparrow = \argmin_{g\in \mathcal{G}_\uparrow} \mathbb{E}\left[(g(p) - \lfdr(p))^2\right].
\]
Although obtained by projection rather than explicit conditioning, the isotonized $\lfdr$ is perfectly calibrated and serves as the benchmark against which monotone calibrators are evaluated throughout the paper. \vspace{.5em}

\begin{theorem}\label{prop-iso-lfdr}
    The isotonized local false discovery rate $\lfdr_\uparrow$ is perfectly calibrated.
\end{theorem}
This result instantiates a classical ``self-consistency" property of isotonic regression, first shown by \citet{brunk1963extension,brunk1965conditional}; see \citet{arnold2025isotonic} for a modern and more general treatment.

Perfect calibration is unattainable in practice, motivating the need for quantitative measures of miscalibration and overall predictive accuracy. 
Many such measures have been proposed, including the expected calibration error \citep{naeini2015obtaining} and alternatives addressing various shortcomings \citep{blasiok2023unifying,blasiok2024smooth,okoroafor2025near,qiao25a,rossellini2025can}.
We take as our primary target the Brier score~\citep{glenn1950verification},
\begin{align}
    \BS{g} 
    = \EE[(Y - g(p))^2] 
    = \underbrace{\EE \left[\left(\EE[Y \mid g(p)] - g(p)\right)^2\right]}_{\text{calibration}} + \underbrace{\mathbb{E}\left[\text{Var}(Y\mid g(p))\right]}_{\text{refinement}},
\end{align}
which simultaneously captures calibration error and the informativeness lost by coarsening.
When predictors are restricted to a class $\mathcal{H}$, we can remove the approximation error intrinsic to~$\mathcal{H}$ with the \emph{Brier regret}
\begin{equation}\label{eq:brier-regret}
    \mathrm{Reg}_\mathcal{H}(g) = \BS{g} - \inf_{h\in \mathcal{H}}\BS{h}.
\end{equation}
For example, when $\mathcal{H}=\mathcal{G}_\uparrow$, the benchmark $\inf_{h\in \mathcal{H}}\BS{h}$ becomes $\BS{\lfdr_\uparrow}$. 
The richer the function class~$\mathcal{H}$, the more demanding the benchmark: vanishing Brier regret with respect to~$\mathcal{G}_\uparrow$ is a much stronger guarantee than vanishing regret with respect to the class of constant predictors, for example.

\subsection{Stochastic supervision via pseudo-labels}\label{sec:pseudo-labels}

The preceding discussion is formally equivalent to supervised classification, except that, crucially, the labels~$Y_i$ are \emph{unobserved} in our setting. 
This is problematic if, for example, one wishes to assess the calibration of a predictor~$g$, since such diagnostics obviously require labeled data. 
In this section, we construct pseudo-labels~$\widetilde{Y}_i$ that are computable from the $p$-values alone and serve as ready substitutes for labels~$Y_i$, allowing much of the existing arsenal of calibration methods to be brought to bear on multiple testing.

Let $p_{(1)}\le \cdots\le p_{(m)}$ denote the ordered $p$-values, with the convention~$p_{(0)} = 0$, and let $R_i$ denote the rank of $p_i$ such that $p_i = p_{(R_i)}$. Given an estimate~$\hpinull$ of~$\pinull$, define
\begin{equation}\label{eq:def-pseudo-label}
        \widetilde{Y}_{i}:= m\hpinull(p_{(R_i)} - p_{(R_i-1)}).
\end{equation}
Unlike the hidden, binary labels~$Y_i$, the pseudo-labels~$\widetilde{Y}_i$ are strictly positive, continuous random variables that may exceed one. Their role is instead to provide an alternative supervised model with the same regression target as the (hidden) binary label, namely~$\lfdr(p_i)$. The following classical limit theorem on the spacings between order statistics can be used to justify this approximation.

\begin{theorem} \citep{pyke1965spacings} 
    Let $F$ (and $f$) denote the marginal cumulative distribution function (and marginal density) of~$p$. Let $(r_m)$ be any integer sequence such that~$r_m/m\to\tau$ as~$m\to\infty$. If~$f$ is continuous and positive in a neighborhood of~$F^{-1}(\tau)$, then 
    \[
    mf(p_{(r_m)})(p_{(r_m)}-p_{(r_m-1)}) \xrightarrow{d} \text{Exponential}(1)
    \]
    as $m\to\infty$. More generally, any fixed collection of spacings with ranks converging to distinct quantiles converges jointly to independent standard exponentials.
\end{theorem}

Consequently, if~$\hpinull\xrightarrow{p} \pinull$, then the pseudo-labels approximately satisfy a regression model with multiplicative errors that are exponentially distributed:
\begin{equation}\label{eq:pseudo-exp-model}
    \widetilde{Y}_i\mid p_i\stackrel\cdot\sim \lfdr(p_i)\,E_i,\qquad E_i\sim\mathrm{Exp}(1),
\end{equation}
such that 
\[
\mathbb{E}\left[\widetilde{Y}_i\mid p_i\right] \approx \lfdr(p_i),
\]
matching the regression function in the binary classification model~\eqref{eq:lfdr}. We emphasize that~\eqref{eq:pseudo-exp-model} is not meant to serve as an exact probabilistic model: the spacings between order statistics are globally constrained (not independent), and the exponential errors are only exact in the limit as $m\to\infty$.

The key property of pseudo-labels is that they can substitute for the true labels for estimating any functional~$\mathbb{E}[\Psi(p, Y)]$ of the joint distribution~$(p, Y)$. 
Since~$Y$ is binary, any such function~$\Psi$ is linear in~$Y$, so it suffices to show that weighted averages of the form~$\frac{1}{m}\sum_{i=1}^m \psi(p_i)\widetilde{Y}_i$ and~$\frac{1}{m}\sum_{i=1}^m \psi(p_i)Y_i$ share the same limit.
Let~$\psi : [0,1]\to\mathbb{R}$ be a Riemann integrable test function. Then, by construction,
\[
    \frac{1}{m}\sum_{i=1}^m \psi(p_i)\widetilde{Y}_i
    = \hpinull \sum_{r=1}^m \psi(p_{(r)})(p_{(r)} - p_{(r-1)})\xrightarrow{p} \pinull\int_0^1 \psi(u)\,\text{d}u.
\]
On the other hand, using the law of large numbers and uniformity under the null~\eqref{eq:uniform-nulls},
\[
    \frac{1}{m}\sum_{i=1}^m \psi(p_i)Y_i = \frac{1}{m}\sum_{i=1}^m \psi(p_i)1\{Y_i=1\} \xrightarrow{p} \pinull \mathbb{E}[\psi(p)\mid Y=1] = \pinull\int_0^1\psi(u)\,\text{d}u.
\]
The limits agree, so the pseudo-labels serve as drop-in replacements for the true labels in standard calibration tools devised for supervised settings, such as reliability diagrams, miscalibration measures, and post-hoc calibration methods.
Of course, the null integral $\int_0^1\psi(u)\,\text{d}u$ could also be computed directly, and this may be advantageous to using~$\frac{1}{m}\sum_{i=1}^m \psi(p_i)\widetilde{Y}_i$. 
The point is that this framing allows us to import the full toolkit of supervised calibration into multiple testing, with pseudo-labels bridging the gap created by hidden labels.

\section{Isotonic calibration and the Grenander estimator} \label{sec:post-hoc}

In this section, we study the monotone post-hoc calibration problem in multiple testing.
The monotonicity constraint encodes the basic principle that smaller $p$-values should provide stronger evidence against the null.
We discuss several ways to enforce this requirement: as isotonic calibration applied to the pseudo-labels, as maximum likelihood in the pseudo-label regression model, or a nonparametric empirical Bayes approach based on estimating the marginal density of the $p$-values.
We find that, despite their different motivations, these three methods are actually equivalent. 
We then establish finite-sample regret bounds for the resulting monotone predictor.

\subsection{Isotonic calibration from pseudo-labels}

The most direct approach to monotone calibration is ordinary isotonic regression on the pseudo-labels~$(\widetilde{Y}_{i})$. The first step is to sort the responses~$(\widetilde{Y}_{i})$ according to the covariate~$(p_i)$. 
With $p_{(r)}$ denoting the $r^\text{th}$ smallest $p$-value, the corresponding pseudo-label is~$\widetilde{Y}_{(r)} = m\hpinull(p_{(r)} - p_{(r-1)})$. 
For the moment, we work with the fitted values at the observed $p$-values. Write $\theta_r = g(p_{(r)})$ for $r = 1,\ldots,m$. A monotone predictor~$g\in \mathcal{G}_\uparrow$ corresponds to a nondecreasing sequence~$\theta_1\le\cdots\le\theta_m$, so the isotonic projection is 
\begin{align}\label{eq:pava}
    \widehat\theta^{\text{BS}}  = \argmin_{0\le \theta_1\le\cdots\le \theta_m\le 1} \frac{1}{m}\sum_{r=1}^m \left(\theta_r-\widetilde{Y}_{(r)}\right)^2.
\end{align}
The projection~$\widehat\theta$ can be computed in $O(m)$ time using the pool adjacent violators algorithm \citep{barlow1972statistical,grotzinger1984projections}. This procedure takes as input the pseudo-labels $\widetilde{Y}_{(1)},\ldots,\widetilde{Y}_{(m)}$, sorted according to $p_{(1)}\le\cdots\le p_{(m)}$, iteratively pooling the values over neighboring blocks whose fitted averages violate monotonicity, reporting the average pseudo-label within each final block.

If consecutive ranks~$g,\ldots,d\in \{1,\ldots,m\}$ are pooled, the fitted value on that block (before truncating at one) is
\begin{align}\label{eq:block-avg}
\widetilde{Y}_{g:d}
= \frac{1}{d-g+1}\sum_{r=g}^d \widetilde{Y}_{(r)}
= \frac{\hpinull m(p_{(d)} - p_{(g-1)})}{d-g+1}.
\end{align}
This expression has a familiar interpretation in multiple testing: it is the canonical plug-in estimate of the false discovery proportion (FDP) among $p$-values in the interval $(p_{(g-1)}, p_{(d)}]$. The numerator $\hpinull m(p_{(d)} - p_{(g-1)})$ estimates the expected number of nulls in that interval under the assumption that nulls are uniformly distributed, while the denominator simply counts the total number of $p$-values in that interval. Isotonic calibration thus pools adjacent $p$-values into intervals on which the local FDP estimate is monotone, and reports that estimate as the calibrated score.

The squared-error criterion should be thought of as a plug-in estimate of the Brier score. From this perspective, $\widehat\theta^{\text{BS}}$ is a straightforward application of the classical approach to isotonic calibration \citep{zadrozny2002transforming}.

\subsection{Maximum likelihood for the pseudo-label regression}

Equation~\eqref{eq:pseudo-exp-model} suggests a working probability model for the pseudo-labels with multiplicative errors that follow a standard exponential distribution. This motivates a likelihood-based approach to estimating the regression function~$\lfdr$. For a predictor~$g$, the pseudolikelihood is $\prod_{i=1}^m \frac{1}{g(p_i)} \exp\left(-\frac{\widetilde{Y}_i}{g(p_i)}\right)$. Up to additive constants not depending on $g$, the negative log pseudolikelihood is
\[\sum_{i=1}^m \left(\frac{\widetilde{Y}_i}{g(p_i)} -\log \frac{\widetilde{Y}_i}{g(p_i)} - 1\right),\]
which is a Bregman divergence known as the \emph{Itakura--Saito distance} \citep{banerjee2005clustering}. The corresponding monotone-constrained estimator is
\begin{align}\label{eq:maximum-pseudolikelihood}
        \widehat\theta^{\text{IS}} = \argmin_{0\le\theta_1\le\cdots\le\theta_m\le 1} \frac{1}{m}\sum_{r=1}^m \left(\frac{\widetilde{Y}_{(r)}}{\theta_r} - \log \frac{\widetilde{Y}_{(r)}}{\theta_r} -1\right).
    \end{align}
Although we focus on constraining the fitted values by the monotone class $\mathcal{G}_\uparrow$, the same likelihood approach could apply to other functional forms for $g$, such as splines or sigmoid functions \citep{platt1999probabilistic}.

\subsection{Empirical Bayes estimation of the lfdr}

Instead of performing regression on the pseudo-labels, the empirical Bayes approach usually starts by observing that, by Bayes' rule,
\[
\lfdr(t) = \frac{\pinull}{f(t)}.
\]
Estimating the $\lfdr$ can thus be reduced to separately estimating the null proportion~$\pinull$ and the marginal density~$f$ of the $p$-values.
Imposing monotonicity of the $\lfdr$ corresponds to requiring~$f$ to be a nonincreasing density on $[0,1]$, and the maximum likelihood estimator of a nonincreasing density is known as the Grenander estimator~\citep{grenander1956theory}:
\[
\widehat f_m \in \argmax_{h\in \mathcal{F}_\downarrow} \frac{1}{m}\sum_{i=1}^m \log h(p_i),
\]
where $\mathcal{F}_\downarrow = \{h : [0,1]\to\mathbb{R}_+ : \int h = 1, h(u)\ge h(v) \text{ for all } u\le v\}$. The Grenander estimator has been extensively studied as a shape-constrained density estimator \citep[see, e.g.,][and references therein]{groeneboom2014nonparametric,samworth2026nonparametric}. The plug-in estimator of the isotonized $\lfdr$ is then
\begin{align}\label{eq:npeb}
    \widehat\theta_r^\text{EB} = \min\left\{\frac{\hpinull}{\widehat f_m(p_{(r)})}, 1\right\}.
\end{align}
This empirical Bayes view connects directly to recent work on lfdr estimation under monotonicity. \citet{strimmer2008unified} first proposed using $\widehat f_m$ in the context of multiple testing. \citet{soloff2024edge} prove finite-sample guarantees for the \emph{support line procedure}, which rejects the null hypotheses corresponding to the $R$ smallest $p$-values, where $R$ is the last index~$k$ such that $\widehat\theta_k^\text{EB}\le \alpha$. 

\subsection{Equivalence}

\begin{theorem}\label{thm:equivalence} Fix any estimator $\hpinull\in (0,1]$ of $\pinull$.
The estimators defined by \eqref{eq:pava}, \eqref{eq:maximum-pseudolikelihood}, and \eqref{eq:npeb} coincide.
\end{theorem}

We define the isotonic calibration estimator~$\widehat\lfdr_\uparrow : [0,1]\to[0,1]$ as the right-continuous step function that is constant between fitted values:
\begin{align}
    \widehat \lfdr_\uparrow(t) &= \begin{cases}
        \widehat{\theta}_1, & 0\le t\le p_{(1)}, \\
        \widehat{\theta}_2, & p_{(1)}< t\le p_{(2)}, \\
        ~\vdots \vspace{.2em}  \\ 
        \widehat{\theta}_m, & p_{(m-1)}< t\le p_{(m)}, \\
        1, & p_{(m)}< t\le 1.
    \end{cases}
\end{align}
Theorem~\ref{thm:equivalence} combines two classical results into a single statement, with different notation for our context. The equivalence between \eqref{eq:pava} and \eqref{eq:maximum-pseudolikelihood} uses the fact that isotonic regression simultaneously minimizes all Bregman divergences---see, e.g., Theorem 1.5.1 of \citet{rdw1988order}. The equivalence between \eqref{eq:pava} and \eqref{eq:npeb} follows from a well-known connection between isotonic regression and the Grenander estimator~$\widehat f_m$---see, e.g., Section~7.2 of \citet{rdw1988order}. The importance of this observation is thus not in its novelty, but in the connection it establishes between a nonparametric empirical Bayes method and performing post-hoc calibration directly on the pseudo-labels.

\subsection{Brier regret guarantee}

We now state a finite-sample guarantee on the Brier regret of $\widehat\lfdr_\uparrow$ over the monotone class $\mathcal{G}_\uparrow$.

\begin{theorem} \label{thm:regret-conv} Suppose $(p_i, Y_i)$ follow an i.i.d. two-groups model with $p_i\mid Y_i=1\sim \text{Unif}([0,1])$. For any $(0,1]$-valued estimator
$\hpinull$ of the null proportion $\pinull$, the Brier regret~\eqref{eq:brier-regret} over~$\mathcal{G}_\uparrow$ satisfies
\[
\EE\left[\mathrm{Reg}_{\mathcal{G}_\uparrow}(\widehat\lfdr_\uparrow)\right]
\le
\sqrt{\frac{2\pi}{m}}
+
2\,\mathbb E\left|\hpinull-\pinull\right|.
\]
In particular, if $\hpinull\xrightarrow{L_1}\pinull$, then isotonic calibration has asymptotically vanishing Brier regret.
\end{theorem}

We emphasize that the bound holds \emph{regardless of whether $\lfdr_\uparrow = \lfdr$;} that is, regardless of whether the true $\lfdr$ is monotone. 
An appealing feature of the result is that the bound separates the error from estimating the null proportion $\pinull$ from the remaining estimation error. A common working assumption in multiple testing is sparsity, meaning $\pinull\approx 1$, often handled by conservatively setting $\hpinull = 1$. Our bound shows that this choice incurs only a bias of $2(1-\pinull)$, which vanishes in the sparse limit $\pinull\to 1$ as $m\to\infty$.

We prove Theorem~\ref{thm:regret-conv} in Appendix~\ref{sec:regret-proof}. 
The bound does not immediately follow from standard empirical process theory: the pseudo-labels are neither independent nor identically distributed. We sidestep this issue by first showing that the isotonic calibration~$\widehat\lfdr_\uparrow$ can in fact be written as the solution to a plug-in estimate of the Brier score that only depends on the $p$-values through their empirical cdf.
The result then follows from standard empirical process theory; in fact $\sqrt{\frac{2\pi}{m}}$ comes directly from the Dvoretzky--Kiefer--Wolfowitz (DKW) inequality on~$\|F_m-F\|_\infty$. The assumption that the $p$-values are i.i.d. is not strictly necessary; for vanishing regret, it instead suffices that the empirical cdf $F_m$ converges uniformly to~$F$ and $\hpinull\xrightarrow{L_1}\pinull$.

\subsection{From calibrated scores to actionable decisions}

The Brier score evaluates the overall accuracy of $g(p)$ as a probability forecast for predicting $Y$. 
In practice, however,
probability forecasts are often used as an intermediate step to making a dichotomous decision, i.e., acting as though $Y=0$ or $Y=1$.
In hypothesis testing, the decision is whether to reject the null hypothesis.

Any decision-maker implicitly faces a tradeoff between false positives (erroneously rejecting the null when it is true) and false negatives (erroneously failing to reject the null when it is false). 
Let $\alpha$ denote the relative cost of a false negative, so that $1-\alpha$ is the relative cost of a false positive. 
Under this asymmetric loss function, the Bayes-optimal decision is to reject the null hypothesis whenever $\lfdr(p)\le \alpha$ \citep{sun2007oracle}. 
This motivates evaluating a predictor $g$ by its weighted classification risk
\begin{align}\label{eq:decision-loss}
    \mathcal{R}_\alpha(g) = \alpha\,\mathbb P\{g(p)>\alpha,\,Y=0\} + (1-\alpha)\,\mathbb P\{g(p)\leq \alpha,\,Y=1\}.
\end{align}
Thus, $\mathcal{R}_\alpha(g)$ evaluates the quality of the decisions obtained by acting on the forecast~$g$ with a specific relative cost $\alpha$. The Brier score aggregates these cost-specific binary decision problems. The two notions are connected via the Schervish representation \citep{savage1971elicitation,schervish1989general}
\[
\BS{g} = 2\int_0^1 \mathcal{R}_\alpha(g)\,\text{d}\alpha.
\]
The Brier regret guarantee of \Cref{thm:regret-conv} therefore controls the excess weighted classification risk of $\widehat{\lfdr}_\uparrow$ with respect to a cost parameter $\alpha$ selected uniformly at random. An individual decision-maker, however, typically faces a particular tradeoff, so it is useful to study the fixed $\alpha$ risk directly.

For any monotone predictor $g\in \mathcal{G}_\uparrow$, the rejection region $\{p : g(p)\le \alpha\}$ is an interval of the form $[0, \tau]$, where $\tau = \sup\{t : g(t)\le \alpha\}$. 
Define the weighted classification risk of the threshold~$\tau$ as
\begin{align*}
    R_\alpha(\tau) 
    &= \alpha (1-\pinull) + \pinull\,\tau -\alpha F(\tau)
\end{align*}
such that $R_\alpha(\tau) = \mathcal{R}_\alpha(g)$. Consequently, the optimal monotone decision rule is obtained by minimizing $R_\alpha(\tau)$ over thresholds $\tau\in [0,1]$. Replacing $F$ and $\pinull$ by their empirical counterparts (and ignoring the additive constant not depending on $t$) yields the empirical risk 
\[
\hat{R}_\alpha(t) = \hpinull\,t - \alpha\, F_m(t).
\]
Selecting the threshold $\hat\tau$ to minimize this empirical risk is exactly the Support Line procedure at level~$\alpha/\hpinull$. \citet{soloff2024edge} show that the selected threshold corresponds to thresholding the Grenander estimator:
\[\hat\tau = \sup\{t : \widehat\lfdr_\uparrow(t)\le \alpha\}.\]
The following theorem shows that $\hat\tau$ has low weighted classification regret over the monotone class~$\mathcal{G}_\uparrow$.

\begin{theorem}\label{thm:actionable}
    Suppose $(p_i, Y_i)$ follow an i.i.d. two-groups model with $p_i\mid Y_i=1\sim \text{Unif}([0,1])$. For any $(0,1]$-valued estimator
    $\hpinull$ of the null proportion $\pinull$, 
    \[
    \mathbb{E}\left[\mathcal{R}_\alpha(\widehat\lfdr_\uparrow)\right] - \inf_{g\in \mathcal{G}_\uparrow}\mathcal{R}_\alpha(g) 
    \le \alpha\sqrt{\frac{2\pi}{m}} + 2\,\mathbb E\left|\hpinull-\pinull\right|.
    \]
\end{theorem}

 This result has parallels with Theorem~6 of \citet{soloff2024edge}, which computes the asymptotic weighted classification regret of the Support Line procedure. While our result achieves a slower rate, the main advantages are that it's valid in finite samples and that we do not assume the true $\lfdr$ is monotone. 
 \Cref{thm:actionable} also complements the Brier regret guarantee of \Cref{thm:regret-conv} by moving from a global evaluation of probability forecasts to the threshold-specific binary decision problem. 
The proof, given in Appendix~\ref{sec:proof-of-thm:actionable}, follows the same overall approach as the proof of Theorem~\ref{thm:regret-conv}. 
As before, the term $\sqrt{\frac{2\pi}{m}}$ in the regret bound comes directly from the DKW inequality on~$\|F_m-F\|_\infty$. Thus, independence is not actually necessary; uniform convergence of the empirical cdf suffices.
One new feature of this result is the factor $\alpha$ multiplying $\sqrt{\frac{2\pi}{m}}$:
the contribution to the regret from estimating~$F$ shrinks with the false-negative cost $\alpha$, and the fixed-threshold problem becomes correspondingly easier as $\alpha \to 0$.

\section{Empirically assessing calibration}
\label{sec:assessment}

The previous section studies methods for \emph{achieving} approximate calibration in multiple testing. 
Here we turn to the dual problem of \emph{assessing} whether a given predictor~$g$ is calibrated on a particular dataset.
In supervised settings, the standard tool is the reliability diagram~\citep{blasiok2024smooth}, which compares the average score within a bin to the empirical frequency of the positive class within that bin. 
For a perfectly calibrated predictor, the resulting points lie on the diagonal $y = x$. 
Reliability diagrams can thus be viewed as a univariate regression of the labels onto the predictions~\citep{brocker2008some, copas1983plotting}, with histogram binning as the simplest such regression. 
The challenge in our setting is, once again, that the labels~$(Y_i)$ are unobserved. We show that the pseudo-labels serve as a valid substitute. We assume throughout that $g$ is a pre-trained predictor.
As is standard in supervised learning, when the same data is used to construct and evaluate the predictor, one should split the data into disjoint train and test sets.

\paragraph{Histogram binning with pseudo-labels.} Fix a predictor $g$ and partition its range into $B$ bins $A_1,\ldots,A_B\subset [0,1]$. The supervised reliability diagram plots, for each bin~$b$, the pair $(\tilde{g}^{(b)}, \pi^{(b)})$, where $\tilde{g}^{(b)} = \EE[g(p) \mid g(p)\in A_b]$ and $\pi^{(b)} = \EE[Y\mid g(p)\in A_b]$ are the average score and null probability within the bin. Up to the discretization, perfect calibration corresponds to $\pi^{(b)}=\tilde{g}^{(b)}$ for every~$b$. 
Because labels are hidden, we cannot estimate $\pi^{(b)}$ directly as an empirical frequency~$\frac{1}{n_b}\sum_{i : g(p_i)\in A_b} Y_i$, but we can use the pseudo-labels $\hat\pi^{(b)} := \frac{1}{n_b}\sum_{i : g(p_i)\in A_b} \widetilde{Y}_i.$ Note that this is the block average~\eqref{eq:block-avg} over the bin~$A_b$. Deviations from the diagonal can be interpreted as in supervised reliability diagrams, with the bin-wise gap $(\hat\pi^{(b)} - \tilde{g}^{(b)})^2$ approximating the calibration component of the Brier score.

\paragraph{Beyond histogram binning.} Histogram binning shares many of the well-known limitations of ECE for assessing calibration error: ECE is not efficiently estimable from samples, and its binned counterpart depends sensitively on the choice of bin width. Kernel-smoothed reliability diagrams~\citep{blasiok2024smooth} address these issues and are also compatible with our pseudo-label approach.
We carry out this construction and apply it to the real-data example in Appendix~\ref{app:rd}.

\section{Experiments} \label{sec:experiments}

We conduct simulation studies and a real-data analysis to illustrate the methods of Sections~\ref{sec:post-hoc} and~\ref{sec:assessment}.

\subsection{Simulation studies} \label{sec:sims}

We compare the Brier regret over~$\mathcal{G}_\uparrow$ of three calibrators: the identity predictor (raw $p$-values), Storey's $q$-value \citep{storey2002direct}, and our isotonic estimator~$\widehat\lfdr_\uparrow$.

\paragraph{Simulation setup.} $p$-values are drawn from a two-groups model~$f = \pinull f_0+(1-\pinull)f_1$, where $f_0$ is the uniform density on~$[0,1]$ and $f_1$ is a~$\text{Beta}(\alpha,\beta)$ density. We vary $\pinull\in\{0.5, 0.75,0.9\}$, $\alpha\in \{0.5, 0.95, 1.5\}$, $m\in \{10^2, 10^3, 10^4, 5\times 10^4\}$ and fix $\beta = 2.3$. 
Values~$\alpha \le 1$ correspond to scenarios where the $\lfdr$ is itself monotone, and~$\alpha=1.5$ violates monotonicity so that~$\lfdr\ne \lfdr_\uparrow$. 

\begin{figure}[t!]
    \centering
    \includegraphics[width=\textwidth]{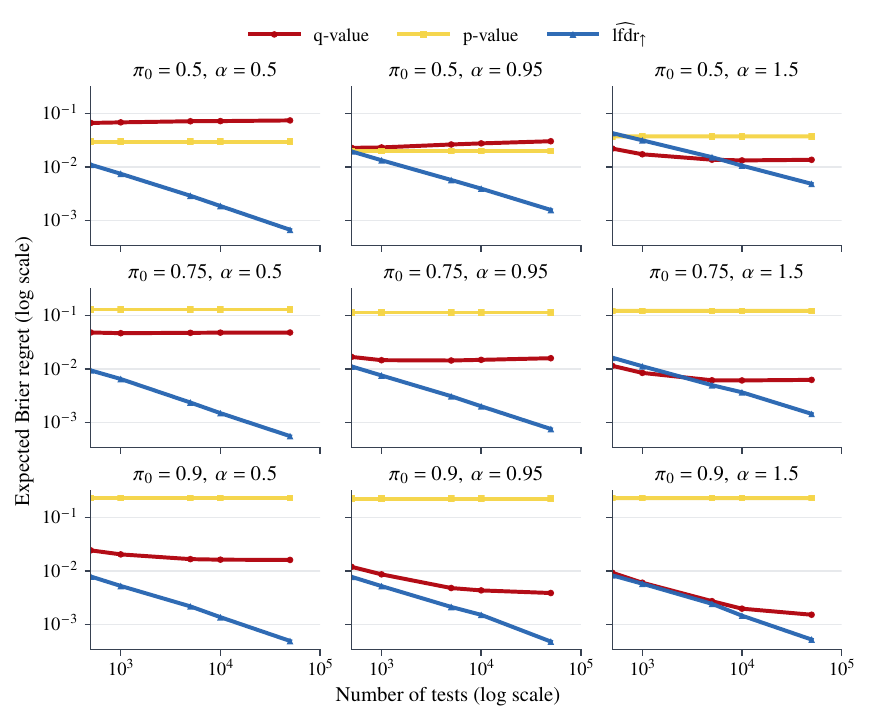}
    \caption{$\log$-$\log$ plots for simulation studies of (inductive) expected Brier regret for $p$-values, $q$-values and the isotonic calibration estimator $\widehat\lfdr_\uparrow$. Each panel represents a new simulation scenario, with different null proportion~$\pinull$ and parameter~$\alpha$ for the alternative density. }
    \label{fig:simulation-studies}
\end{figure}

\paragraph{Evaluation.} For any piecewise constant calibrator $\hat{g} \in \{\widehat \lfdr_\uparrow, \hat{q} \}$, the Brier regret~\eqref{eq:brier-regret} can be evaluated as
\begin{align}
    \mathrm{Reg}_\mathcal{G_\uparrow}(\hat{g}) &= \sum_{s \in \text{level sets of } \hat g} \int_{p : \hat{g}(p)=s} f(t) (\lfdr(t) - s)^2 dt - \EE[(\lfdr(p)  - \lfdr_{\uparrow}(p))^2] 
    \label{eq:level-sets}
\end{align}
which we evaluate by numerical integration. We report average regret over~$500$ Monte Carlo draws. 

\paragraph{Results.} Figure~\ref{fig:simulation-studies} displays the results across simulation scenarios. The regret of $\widehat \lfdr_\uparrow$ decreases steadily with the number of tests $m$ in every panel.
By contrast, the regret for $p$-values is flat (since the predictor is the identity and thus does not depend on~$m$), and the regret for $q$-values is larger than that of $\widehat \lfdr_\uparrow$ in most cases. 
This contrast reflects what Brier regret over $\mathcal{G}_\uparrow$ is measuring: the gap to the best attainable Brier score by any monotone predictor. 
The $p$- and $q$-value predictors do not target the isotonized $\lfdr$, and so their regret does not vanish even as $m$ diverges.
In line with Theorem~\ref{thm:regret-conv}, $\widehat\lfdr_\uparrow$ converges at a polynomial rate.
The estimated rates for the Brier regret of $\widehat\lfdr_\uparrow$ are provided in Table~\ref{tab:decay_rates}.

The rightmost panel of Figure~\ref{fig:simulation-studies} provides results on the case when we apply our post-hoc calibration method when the monotonicity assumption is violated. Here, we still see decay in the regret of $\widehat \lfdr_\uparrow$. This too is consistent with the theoretical analysis: unlike the first two vertical panels, $\lfdr \neq \lfdr_\uparrow.$ However, $\lfdr_\uparrow$ is perfectly calibrated, and by Theorem~\ref{thm:regret-conv}, $\widehat \lfdr_\uparrow$ still converges to $\lfdr_\uparrow.$

\subsection{Real data application} \label{sec:rd}
We apply the post-hoc calibration method to data curated by \citet{szucs2017empirical}, comprising $m \approx 27,000$ records of $t$-statistics and degrees of freedom scraped from 18 journals that published work in cognitive neuroscience and experimental psychology (2011-2014). Figure~\ref{fig:real-data} (left) shows the empirical distribution of the resulting $p$-values. The estimated density appears monotone decreasing, suggesting that the working assumption underlying isotonic calibration is well-suited to these data. 

Szucs and Ioannidis sought to estimate the probability that a statistically significant finding is false, which they dub the ``false reporting rate" (FRP). They compute FRP by the formula $\mathrm{FRP} = \frac{O \alpha }{O \alpha + \text{Power}}$, where $O$ is the pre-study odds of the null hypothesis, and $\alpha$ is often $0.05.$ Power was calculated from the non-central $t$-distribution. The authors present a sensitivity analysis to show that their calculations for $\mathrm{FRP}$ vary substantially for different plausible prior odds.

We use Storey's estimator $\hpinull$ and use the histogram regression method presented in Section~\ref{sec:assessment} to assess calibration, splitting the data into training and test sets. 
From Figure~\ref{fig:real-data} (right), the~$p$-values and $q$-values appear severely miscalibrated, while the points for~$\widehat\lfdr_\uparrow$ lie close to the diagonal. 
Practitioners may thus reasonably interpret the calibrated values as local probability statements for individual studies. 
The smoothed reliability diagram in Appendix~\ref{app:rd} shows the same qualitative pattern, suggesting some robustness to bin choice. 

While \citet{szucs2017empirical} seek to estimate a `post-hoc' false positive rate, their method still is not adapted to make local probability claims about the individual hypotheses. From a meta-science perspective, knowing which findings within a given subfield have a high probability of being true or not true is of great interest. Our method is particularly well-adapted to the meta-science goals, since it is more reasonable to assume that aggregated test statistics from many studies in a \emph{discipline} are more likely to be independent and identically distributed from a two-groups model than, say, test statistics aggregated from a genome-wide association study.

\begin{figure}[t!]
    \centering
        \includegraphics{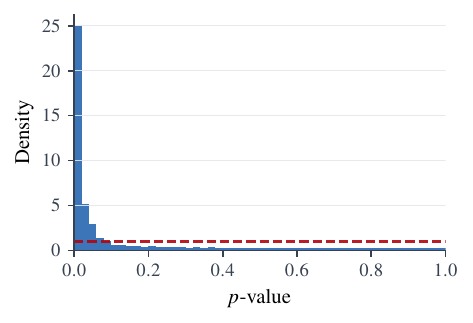}
        \includegraphics{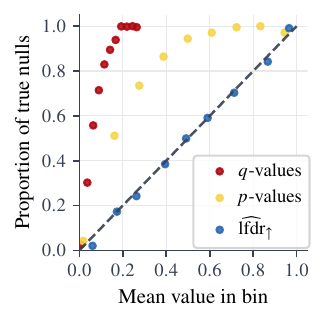}
    \caption{
    Calibration results on data curated by \citet{szucs2017empirical}. Left: Empirical distribution of $p$-values. The dashed line denotes the $\mathrm{Unif}(0,1)$ density. Right: Reliability plot assessing calibration of reported $p$-values.
    }
    \label{fig:real-data}
\end{figure}

\section{Discussion}

We have argued that calibration is a key inferential goal of large-scale hypothesis testing and shown how our pseudo-label construction enables calibration without ground-truth labels.
Our approach, through the use of pseudo-labels constructed from test statistics themselves, draws together three natural viewpoints for treating many related tests as a resource: isotonic calibration, MLE regression on the pseudo-labels and nonparametric empirical Bayes estimation of the lfdr over the class of monotone functions. We show the equivalence between these three viewpoints in that they yield the same calibrator, which has vanishing Brier regret over~$\mathcal{G}_\uparrow$ as the number of tests $m \to \infty$. The first two views do not assume any model for the alternative densities. This highlights that informative approximate calibration through the means of the third approach (empirical Bayes) is consistent, even when the monotonicity assumption for the densities does not hold. 

Turning attention to in-sample assessment of calibration, we show how any functional $\EE[\Psi(p, Y)]$ can be estimated consistently by substituting the pseudo-labels for $Y$. This allows the suite of tools for assessing calibrators in supervised settings to be used in multiple testing, as the number of tests grows, including popular options such as histogram binning and its smoothed counterpart. Our real data illustrations provide a principled test bed for making calibrated claims about individual hypotheses in large, meta-analysis data sets~\citep{patel2026atlas, ioannidis2008calibration}.

\paragraph{Limitations and future directions.} There are several limitations to our work which inspire our future directions. First, the theory presented in this paper is centered around using $p$-values as test statistics. The ability to reasonably use $p$-value gaps as pseudo-labels hinges on their uniformity under the null. We pose it as an open question for how to construct similar pseudo-labels when the class-conditional null density is not uniform, for example with $z$-scores or high-dimensional test statistics. In that vein, connections to estimating the empirical null distribution may be explored. Within our framing there are other connections to supervised learning that can be explored, including multi-class assignments. To further paint the analogy between supervised learning and multiple testing, it is of interest to see how approaches commonly deployed in supervised learning such as cross-validation interface with our methods.

\bibliographystyle{dcu}
\bibliography{refs}

\appendix

\section{Proofs of technical results}

\subsection{Proof of Theorem~\ref{prop-iso-lfdr}}\label{sec:lfdr-iso-cal}

\begin{proof}
We want to show 
\begin{align}\label{eq:calibration-wts}
    \EE\left[\lfdr(p)\mid  \lfdr_\uparrow(p)\right] = \lfdr_\uparrow(p)
\end{align}
almost surely. The left-hand side is the almost surely unique random variable $W$ such that
\begin{itemize}
    \item $W$ is a measurable function of $\lfdr_\uparrow(p)$, and 
    \item $\EE[\lfdr(p)1_A] = \EE[W1_A]$ for all $A\in \sigma(\lfdr_\uparrow(p))$.
\end{itemize}
We claim $W=\lfdr_\uparrow(p)$ satisfies both properties. Measurability is immediate. The second property was shown by~\cite{brunk1965conditional}---see Theorem~3.1. Since $W$ is a.s. unique, we conclude \eqref{eq:calibration-wts}.
\end{proof}

\subsection{Proof of Theorem~\ref{thm:regret-conv}}\label{sec:regret-proof}

The proof of Theorem~\ref{thm:regret-conv} relies on the observation that the optimization objective for our isotonic calibration estimator~$\widehat\lfdr_\uparrow$ can be written in a way that only depends on the data via the empirical cdf
$$F_m(t) = \frac{1}{m}\sum_{i=1}^m1\{p_i\le t\}$$
and \emph{not} on the pseudo-labels $(\widetilde{Y}_i)$.

\begin{lemma}\label{lem:empirical-process-representation}
    The isotonic calibration estimator~$\widehat\lfdr_\uparrow$ minimizes
$$
    \int g^2(p)\,\textnormal{d}F_m(p) - 2\hpinull\int g(p)\,\textnormal{d}p
$$
    over $g\in \mathcal{G}_\uparrow$.
\end{lemma}

\begin{proof}
    By definition, $\widehat\lfdr_\uparrow$ minimizes 
$$\frac{1}{m}\sum_{i=1}^m (g(p_i) - \widetilde{Y}_i)^2$$
over $g\in \mathcal{G}_\uparrow$. Expanding the square, we find that, ignoring constants not depending on $g$, the objective can be rewritten as
$$
\int g^2 \text{d}F_m - \frac{2}{m}\sum_{i=1}^m \widetilde{Y}_i g(p_i) 
= \int g^2 \text{d}F_m - 2\hpinull\sum_{r=1}^m (p_{(r)}-p_{(r-1)}) g(p_{(r)}).
$$
Note that, since $\widehat\lfdr_\uparrow$ is constant except at the observed $p$-values and since $\widehat\lfdr_\uparrow(1)=1$, the integral equals the Riemann sum plus a small correction:
$$
\int \widehat\lfdr_\uparrow\,\text{d}p = \sum_{r=1}^m \widehat\lfdr_\uparrow(p_{(r)})(p_{(r)}-p_{(r-1)}) + (1-p_{(m)}).
$$
On the other hand, for any other $g\in \mathcal{G}_\uparrow$, since $g$ is nondecreasing and $g(1)\le 1$, we have 
$$
\int g\,\text{d}p 
\le \sum_{r=1}^m g(p_{(r)})(p_{(r)}-p_{(r-1)}) + (1-p_{(m)}).
$$
Thus, we have
\begin{align*}
    \int \widehat\lfdr_\uparrow^2\,\text{d}F_m - 2\hpinull\int\widehat\lfdr_\uparrow\,\text{d}p 
    &=\int \widehat\lfdr_\uparrow^2 \,\text{d}F_m - 2\hpinull\left(\sum_{r=1}^m (p_{(r)}-p_{(r-1)}) \widehat\lfdr_\uparrow(p_{(r)}) + (1-p_{(m)}) \right)\\
    &\le\int g^2\, \text{d}F_m - 2\hpinull\left(\sum_{r=1}^m (p_{(r)}-p_{(r-1)}) g(p_{(r)}) + (1-p_{(m)}) \right)\\
    &\le\int g^2\, \text{d}F_m - 2\hpinull\int g\,\text{d}p, 
\end{align*}
proving the lemma.
\end{proof}

\begin{proof}[Proof of Theorem~\ref{thm:regret-conv}]
    Observe 
    \begin{align*}
        \BS{g} 
        &= \EE[Y] + \EE[g^2(p)] - 2\pinull\EE[g(p)\mid Y=1]\\
        &= \pinull + \int g^2(p)\, \text{d}F(p) - 2\pinull \int_0^1 g(p)\,\text{d}p.
    \end{align*}
    Thus,
    \begin{align*}
        \mathrm{Reg}_{\mathcal{G}_\uparrow}(\widehat\lfdr_\uparrow)
        &=\BS{\widehat\lfdr_\uparrow} - \BS{\lfdr_\uparrow}  \\
        &= \int \widehat\lfdr_\uparrow^2\, \text{d}F - \int \lfdr_\uparrow^2\, \text{d}F + 2\pinull\left(\int\lfdr_\uparrow\,\text{d}p-\int\widehat\lfdr_\uparrow\,\text{d}p\right) \\
        &= \int \widehat\lfdr_\uparrow^2\, \text{d}F_m - \int \lfdr_\uparrow^2\, \text{d}F_m +2\hpinull\left(\int\lfdr_\uparrow\,\text{d}p-\int\widehat\lfdr_\uparrow\,\text{d}p\right) \\
        &~~~+\int \widehat\lfdr_\uparrow^2\, \text{d}(F-F_m) - \int \lfdr_\uparrow^2\, \text{d}(F-F_m) \\
        &~~~+ 2(\pinull-\hpinull)\left(\int\lfdr_\uparrow\,\text{d}p-\int\widehat\lfdr_\uparrow\,\text{d}p\right).
    \end{align*}
    The first term is $\le 0$ by Lemma~\ref{lem:empirical-process-representation}. Using $\int g\,\text{d}p \in [0,1]$ and monotonicity of both $\widehat\lfdr_\uparrow$ and $\lfdr_\uparrow$,
    \begin{align*}
        \mathrm{Reg}_{\mathcal{G}_\uparrow}(\widehat\lfdr_\uparrow)&\le 2\sup_{g\in \mathcal{G}_\uparrow} \left|\int g^2\,\text{d}(F_m-F)\right| + 2|\hpinull-\pinull|
    \end{align*}
    Now for any $g\in \mathcal{G}_\uparrow$,
    \[\int g^2\,\text{d}(F_m-F) = \int (F_m-F)\,\text{d}(g^2) \le \|F_m-F\|_\infty \int \text{d}(g^2) \le \|F_m-F\|_\infty.\] 
    In expectation,
    \[
    \EE\left[\mathrm{Reg}_{\mathcal{G}_\uparrow}(\widehat\lfdr_\uparrow)\right] \le 2\EE\|F_m-F\|_\infty + 2|\hpinull-\pinull|.
    \]
    By the DKW inequality \citep{massart1990tight}, 
    \[
    \mathbb{P}\{\|F_m-F\|_\infty\ge t\}\le 2e^{-2mt^2}.
    \]
    Integrating the tail gives
    \[
    \EE\|F_m-F\|_\infty
    \le \int_0^1 2e^{-2mt^2}\text{d}t
    = \sqrt{\frac{2\pi}{m}}\int_0^{2\sqrt{m}} \frac{1}{\sqrt{2\pi}}e^{-s^2/2}\text{d}s
    \le \sqrt{\frac{\pi}{2m}},
    \]
    completing the proof.
\end{proof}

\subsection{Proof of \Cref{thm:actionable}}\label{sec:proof-of-thm:actionable}

\begin{proof}[Proof of \Cref{thm:actionable}]
    Define $\widehat R_\alpha(t)=\alpha(1-\pi_{\mathrm{null}})+\widehat\pi_{\mathrm{null}}\,t-\alpha F_m(t)$;
the added constant does not depend on~$t$, so we
still have $\widehat R_\alpha(\widehat\tau)\le\widehat R_\alpha(t)$ for all $t$,
in particular at $t=\tau^\ast$. Decomposing the regret and using this
optimality,
\begin{align*}
  R_\alpha(\widehat\tau)-R_\alpha(\tau^\ast)
  &= \big[R_\alpha(\widehat\tau)-\widehat R_\alpha(\widehat\tau)\big]
    +\underbrace{\big[\widehat R_\alpha(\widehat\tau)-\widehat R_\alpha(\tau^\ast)\big]}_{\le 0}
    +\big[\widehat R_\alpha(\tau^\ast)-R_\alpha(\tau^\ast)\big]\\
  &\le 2\sup_{t\in[0,1]}\big|R_\alpha(t)-\widehat R_\alpha(t)\big|.
\end{align*}
The constant $\alpha(1-\pi_{\mathrm{null}})$ cancels in the difference $R_\alpha(t)-\widehat R_\alpha(t)$, leaving
\[
  R_\alpha(t)-\widehat R_\alpha(t)
  \;=\;(\pi_{\mathrm{null}}-\widehat\pi_{\mathrm{null}})\,t-\alpha\big(F(t)-F_m(t)\big),
\]
so, using $t\le 1$,
\[
  \sup_{t\in[0,1]}\big|R_\alpha(t)-\widehat R_\alpha(t)\big|
  \;\le\; |\widehat\pi_{\mathrm{null}}-\pi_{\mathrm{null}}|
   \;+\; \alpha\,\|F_m-F\|_\infty .
\]
Taking expectations,
\[
  \E\left[\mathcal{R}_\alpha(\widehat{\lfdr}_\uparrow)\right]-\inf_{g\in\mathcal G_\uparrow}\mathcal{R}_\alpha(g)
  \;\le\; 2\,\E|\widehat\pi_{\mathrm{null}}-\pi_{\mathrm{null}}|
   \;+\; 2\alpha\,\E\|F_m-F\|_\infty .
\]
The rest of the proof proceeds as in the proof of \Cref{thm:regret-conv}.
\end{proof}

\section{Additional experimental details}

\subsection{Additional information for Section~\ref{sec:sims}} \label{app:sims}
\paragraph{Compute details.} All experiments were run on the university computing cluster. We parallelized each Monte Carlo run over 500 scheduled jobs. Each job was run on a standard cluster CPU with 16G memory. The runtime for each job after shuffling the Monte Carlo runs was~4 hours. 

\begin{figure}[t]
    \centering
    \includegraphics{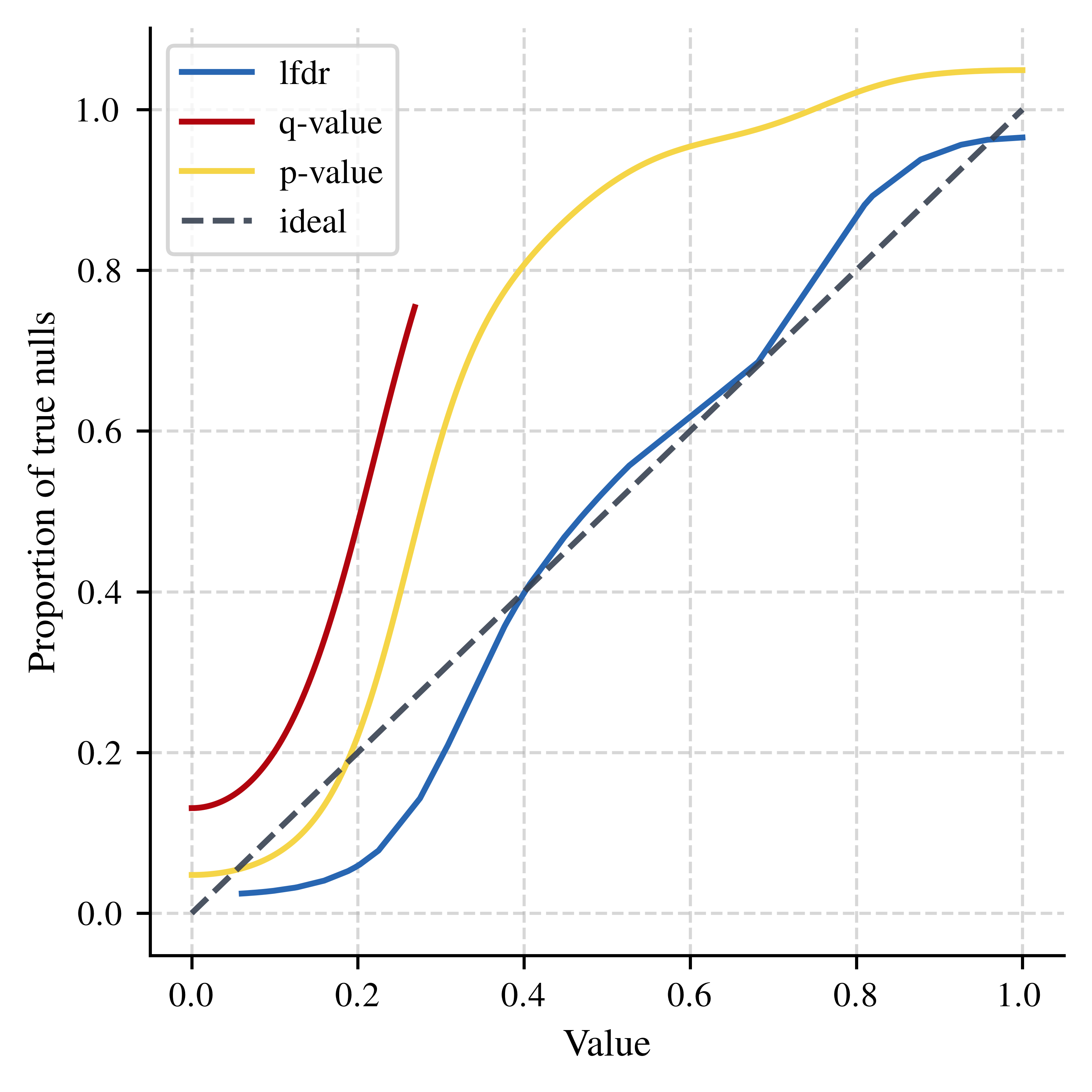}
    \caption{Smoothed reliability plot following the method of~\cite{blasiok2024smooth}. Bandwidth $\sigma = 0.1$.}
    \label{fig:smoothed-reliability}
\end{figure}

\subsection{Additional information for Section~\ref{sec:rd}} \label{app:rd}

\begin{table*}[t]
\centering
\small

\begin{subtable}{\textwidth}
\centering
\begin{tabular}{cccc}
\toprule
$\pinull$ & $\alpha = 0.50$ & $\alpha = 0.95$ & $\alpha = 1.50$ \\
\midrule
0.50 & $-0.6095 \; (0.0088)$ & $-0.5457 \; (0.0056)$ & $-0.4753 \; (0.0069)$ \\
0.75 & $-0.6214 \; (0.0072)$ & $-0.5849 \; (0.0085)$ & $-0.5202 \; (0.0111)$ \\
0.90 & $-0.6004 \; (0.0136)$ & $-0.5963 \; (0.0235)$ & $-0.6012 \; (0.0202)$ \\
\bottomrule
\end{tabular}
\caption{$\widehat \lfdr_\uparrow$}
\end{subtable}

\vspace{1em}

\begin{subtable}{\textwidth}
\centering
\begin{tabular}{cccc}
\toprule
$\pinull$ & $\alpha = 0.50$ & $\alpha = 0.95$ & $\alpha = 1.50$ \\
\midrule
0.50 & $0.0000$ & $0.0000$ & $0.0001 \; (0.0000)$ \\
0.75 & $0.0000$ & $0.0000$ & $0.0000 \; (0.0000)$ \\
0.90 & $0.0000$ & $0.0000$ & $0.0000 \; (0.0000)$ \\
\bottomrule
\end{tabular}
\caption{$p$-value}
\end{subtable}

\vspace{1em}

\begin{subtable}{\textwidth}
\centering
\begin{tabular}{cccc}
\toprule
$\pinull$ & $\alpha = 0.50$ & $\alpha = 0.95$ & $\alpha = 1.50$ \\
\midrule
0.50 & $\phantom{-}0.0235 \; (0.0025)$ & $\phantom{-}0.0652 \; (0.0034)$ & $-0.1026 \; (0.0355)$ \\
0.75 & $\phantom{-}0.0023 \; (0.0035)$ & $-0.0057 \; (0.0206)$ & $-0.1288 \; (0.0460)$ \\
0.90 & $-0.0890 \; (0.0248)$ & $-0.2522 \; (0.0501)$ & $-0.4053 \; (0.0541)$ \\
\bottomrule
\end{tabular}
\caption{$q$-value}
\end{subtable}

\caption{Estimated decay rates with Monte Carlo standard errors (in parentheses) across different calibrators and parameter settings $(\pinull, \alpha)$.}
\label{tab:decay_rates}
\end{table*}

Here, we assess the calibration of our method using more principled approaches that treat the points used to assess ECE through binning as coarsenings from a function that ought to be (smoothly) estimated through a regression problem. We adapt the smoothed reliability diagrams of~\cite{blasiok2024smooth}. The so-called ``calibration function'' is
\begin{align*}
    \mu(t) := \EE[Y \mid g(p) = t]
\end{align*}
which we seek to estimate from the data. Observe that $g$ is perfectly calibrated if $\mu(t) = t$. 
Rather than use histogram binning/regression,~\cite{blasiok2024smooth} use kernel smoothing and estimate $\mu(t):= \EE[Y \mid g(p) = t]$ as
\begin{align*}
    \hat \mu(t) 
    = \frac{\sum_{i} K_\sigma (t, g(p_i)) Y_i}{\sum_i K_\sigma(t, g(p_i))} 
    = \frac{\frac{1}{m}\sum_{i} K_\sigma (t, g(p_i)) Y_i}{\frac{1}{m}\sum_i K_\sigma(t, g(p_i))}
\end{align*}
where $K_\sigma$ is a specially reflected Gaussian kernel. Since we don't observe $Y_i$, we instead propose using the pseudo-labels $\widetilde Y_i$ in their place. Observe that for fixed $t$, $K_\sigma(t, g(p))$ can be represented as a measurable test function $\psi: [0, 1] \to \mathbb{R}$. Then, using the argument presented in Section~\ref{sec:framework} for the numerator, and by the continuous mapping theorem, the ratio converges in probability to
\begin{align*}
   \tilde \mu(t) := \frac{
    \pinull \,\EE\!\left[K_\sigma(t,g(p)) \mid Y = 1\right]
    }{
    \EE\!\left[K_\sigma(t,g(p))\right]
    },
\end{align*}
the population smoothed estimate of $\mu(t).$  Therefore, we can use the pseudo-labels when constructing the smoothed reliability diagrams which plot the pairs $(t, \hat \mu(t)).$

We display the smoothed reliability diagram in Figure~\ref{fig:smoothed-reliability}, and note that it mostly agrees with the histogram binning version of the reliability plot displayed in Figure~\ref{fig:real-data}.

\end{document}